\documentclass{article}
\usepackage{amsfonts}

\usepackage{amsmath,amssymb}
\usepackage{graphicx,color}
\usepackage{epstopdf}
\usepackage{epsfig,nopageno}
\usepackage{mathrsfs}
\usepackage[breaklinks=true]{hyperref}
\newcommand{\beq}{\begin{equation}}
\newcommand{\eeq}{\end{equation}}
\newcommand{\bqa}{\begin{eqnarray}}
\newcommand{\eqa}{\end{eqnarray}}

\definecolor{green}{rgb}{0.00,0.50,0.00}

\newtheorem{theorem}{Theorem}

\newtheorem{corollary}[theorem]{Corollary}

\newtheorem{definition}[theorem]{Definition}

\newtheorem{lemma}[theorem]{Lemma}

\newtheorem{proposition}[theorem]{Proposition}
\newtheorem{remark}[theorem]{Remark}

\newenvironment{proof}[1][Proof]{\noindent\textbf{#1.} }{\ \rule{0.5em}{0.5em}}

\begin{document}

\title{Quantum Filtering for Squeezed Noise Inputs}

 \author{John Gough$^{\dag, \sharp}$, Dylon Rees$^\ast$\\
 Department of Physics, Aberystwyth University, Wales, UK.\\
 $\dag$ Orcid ID: 0000-0002-1374-328X\\
 $\sharp$ jug@aber.ac.uk\\
 $\ast$ dyr6@aber.ac.uk}
 
\date{\today}
\maketitle

\begin{abstract}
We derive the quantum filter for a quantum open system undergoing quadrature measurements (homodyning) where the input field is in a general quasi-free state. This extends previous work for thermal input noise and allows for squeezed inputs.  We introduce a convenient class of Bogoliubov transformations which we refer to as balanced and formulate the quantum stochastic model with squeezed noise as an Araki-Woods type representation. We make an essential use of the Tomita-Takesaki theory to construct the commutant of the C*-algebra describing the inputs and obtain the filtering equations using the quantum reference probability technique. The derived quantum filter must be independent of the choice of representation and this is achieved by fixing an independent quadrature in the commutant algebra.
\end{abstract}

\section{Introduction}
The problem of quantum estimation comes essentially down to giving an optimal observable $\widehat X$ constructed from the measured observables to approximate a given compatible observable $X$. The measurements are naturally required to be compatible (that is, they form a commutative algebra $\mathfrak{M}$) and the observables that we may estimate must likewise be compatible with the measurements (and so belong to the commutant $\mathfrak{M}'$. In quantum filtering problems, the measurement algebra is growing over time as new measurements are made, thereby making the problem a dynamic one \cite{Belavkin,Slava}. Optimality is formulated in a least squares sense \cite{BvHJ,BvH}. Traditionally, this problem has been formulated for quantum open systems where a fixed system interacts with driving (bosonic) vacuum noise processes \cite{HP84} and quadrature measurements (homodyning) are made on the output noise processes \cite{GC85}. However, coherent state input \cite{GK}, multi-photon inputs and superposition of coherent states (Schr\"{o}dinger cat) inputs \cite{GJNC} have also been shown to be amenable.

More recently, the quantum filtering problem has been formulated and solved for thermal input fields \cite{Gough25}. In this case, the input fields are not unitarily equivalent to vacuum fields \cite{HL85}, and one also cannot obtain them through an embedding or a dilation step. Instead, one had to rely on the full Tomita-Takesaki theory \cite{TT} where an Araki-Woods representation \cite{ArakiWoods} was used for the thermal inputs and where the commutant algebra was more involved than in the standard Weyl-Fock setting.

The purpose of this paper is to extend this result to bosonic input processes in general quasi-free states. We are in the fortunate position that a variant of the Araki-Woods technique applies here. In fact, one can simplify the machinery of Tomita-Takesaki in the spirit of Rieffel and van Daele \cite{RvD} and consider this in terms of projections on to symplectically complementary subspaces at the one-particle level which is then second-quantized. This route has been successfully exploited by Belton, Lindsay and coauthors \cite{Belton}.

We will use the underlying concepts of the C*-algebra of the canonical commutation relations (CCR), Bogoliubov transformations, etc., see \cite{Bhat,Kupsch,Petz}. However, the basic formulation is understandable in terms of more familiar concepts from quantum optics, in particular, two-mode squeezing and we make use of standard results on boson modes in a finite-dimensional setting \cite{Bishop}.

The paper is organized as follows. In Section \ref{sec:2mode}, we recall the features behind two-mode squeezing. Here, we introduce the definition of \textit{balanced Bogoliubov transformations} which are a convenient class. An essential role is played by the Hellmich-Honegger-K\"{o}stler-K\"{u}mmerer-Rieckers representation \cite{HKKR} for squeezed noise, however, we use the simple boson mode setting as an accessible way to segue into their results and we end up using the same algebraic coefficients. In Section \ref{sec:symp}, we describe the infinite-dimensional setting, and in \ref{sec:nonfock} we describe the quantum stochastic models for squeezed noise. Finally, in Section \ref{sec:filter} we solve the filtering problem for homodyne measurement of the output quadrature from a squeezed noise input.

Here we have an open quantum system described by a unitary process $U_t$ satisfying a quantum stochastic differential equation \cite{HL85} driven by squeezed bosonic input processes $B_t,B_t^\ast$. The output processes $B_{\text{out}}(t) = U_t^\ast B(t) U_t$ give the scattered fields and we measure the quadrature process $Y(t)=B_{\text{out}}(t)+B_{\text{out}}(t)^\ast$. The filtered estimate for $U_t^\ast X\, U_t$, for any system observable $X$, should belong to the measurement algebra $\mathfrak{Y}_t$ up to time $t$. In fact, $U_t^\ast X\, U_t$ belongs to the commutant of the measurement algebra $\mathfrak{Y}_t'$, and so is compatible with the measurements. We have that $\mathfrak{Y}_t =U_t^\ast \mathfrak{Z}_t U_t$ where $\mathfrak{Z}_t$ is the commutative algebra generated by the \textit{input} quadratures $B_s+B_s^\ast$ for $s\le t$. In the present case, the commutant of $\mathfrak{Z}_t$ includes not just the system but the commuting processes $B_t',B_t^{'\ast}$. The mathematical step to compute the filtered estimate for $U_t^\ast X\, U_t$ is to use the reference probability technique \cite{BvH} to replace $U_t$ with a process $V_t\in \mathfrak{Z}_t'$, see the Kallianpur-Striebel formula in Theorem \ref{thm:QKS}. This is where we encounter the additional commuting processes $B_t',B_t^{'\ast}$ however, by a result of Holevo \cite{Holevo}, we only need to know there action on the cyclic state of the representation (in this case the joint Fock vacuum) in order to project back from the commutant into the measured algebra itself.


\section{Two-mode Squeezing}
\label{sec:2mode}
We take $\mathfrak{F}$ to denote the Fock space of a single boson mode with annihilator $a$ satisfying the CCR $[a,a^\ast ] =1$. A state $\langle \cdot \rangle$ is mean-zero if $\langle a \rangle =0$. Its second moments may be assembled into a covariance matrix
\begin{eqnarray}
C= 
\left[ \begin{array}{cc}
     \langle a a^\ast \rangle &  \langle aa \rangle\\
     \langle a^\ast a^\ast \rangle &  \langle a^\ast a \rangle
\end{array}
\right] 
=
\left[ \begin{array}{cc}
     n+1 &  m\\
     m^\ast   &   n
\end{array}
\right] .
\label{eq:C}
\end{eqnarray}
The vacuum state ($n,m=0$) has covariance matrix $C_{\text{vac}}= \begin{bmatrix}
    1 &0 \\
    0 & 0
\end{bmatrix}$.

The parameter $n=\langle a^\ast a \rangle$ is positive and gives the \textit{average occupation number}. The parameter $m=\langle a^2 \rangle $ is complex and is called the \textit{squeezing}. The eigenvalues of $C$ are
\begin{eqnarray}
    \lambda_\pm = \frac{1}{2}(2n+1) \pm \sqrt{ |m|^2 + \frac{1}{4} } 
\end{eqnarray}
and positivity of $C$ implies
\begin{eqnarray}
    \Delta = \text{det} (C)=n(n+1)- |m|^2 \ge 0.
\end{eqnarray}
We say that the state is mean-zero Gaussian if we have
\begin{eqnarray}
    \langle e^{iu^\ast a +i u\, a^\ast }\rangle = \exp \bigg\{ 
    -\frac{1}{2}[ u^\ast ,u] \, C \left[ \begin{array}{c}
     u\\
     u^\ast
\end{array}
\right]
\bigg\} 
=
e^{- \frac{1}{2} (2n+1) |u|^2-\frac{1}{2}m \, u^{\ast 2} -\frac{1}{2}m^\ast u^2}.
\label{eq:char}
\end{eqnarray}

Gaussian states are also known as \textit{quasi-free}. If $m=0$, then the moment generating function (\ref{eq:char}) only depends on $|u|$ and is said to be \textit{gauge-invariant} (since the replacement $u \to e^{i \theta} u$ leaves all moments unchanged!) and these states are referred to as thermal states.

For $m \neq 0$, the state is then said to be squeezed. We say that the squeezing is \textit{maximal} if we have $\Delta =0$, i.e., the equality $|m|^2 = n(n+1)$, and \textit{sub-maximal} if $|m|^2<n(n+1)$.

    \bigskip
    
    \textbf{Example (Maximal squeezing):}
In the maximal case, the covariance matrix $C$ is singular: it has an eigenvalue $2n+1$ with eigenvector $[\sqrt{n+1}, \sqrt{n}]^\top$ while the other eigenvalue is zero. 

It suffices to just consider one mode: let $a$ be in the vacuum state, then setting
\begin{eqnarray}
    b= e^{i\theta /2} \sqrt{n+1} \, a + e^{-i \theta /2} \sqrt{n}\, a^\ast
\end{eqnarray}
we see that $b$ is also a boson mode with Gaussian state with number $n$ and squeezing $m=e^{i\theta} \sqrt{n(n+1)}$.

\bigskip

Let $a_1,a_2$ be copies of a boson mode, then we may consider the pair as commuting boson modes $a_1 \otimes I$ and $I\otimes a_2$ on $\mathfrak{F} \otimes \mathfrak{F}$. For convenience, we shall drop the tensor product symbols and just write $a_1$ for $a_1 \otimes I$, etc. We will be interested in linear transformations of the form
\begin{eqnarray}
    b_1 &=& x_1 \, a_1 + y_1 \, a_1^\ast + z_1 \, a_2 + w_1 \, a_2^\ast , \nonumber \\
    b_2 &=& x_2 \, a_1 + y_2 \, a_1^\ast + z_2 \, a_2 + w_2 \, a_2^\ast ,
    \label{eq:2mode}
\end{eqnarray}
or equivalently,
\begin{eqnarray}
    \begin{bmatrix}
        b_1 \\
        b_1^\ast \\
        b_2 \\
        b_2^\ast
    \end{bmatrix}
    = S \,
     \begin{bmatrix}
        a_1 \\
        a_1^\ast \\
        a_2 \\
        a_2^\ast
    \end{bmatrix}, \quad
    \text{with}\, \,
    S= \begin{bmatrix}
        x_1 & y_1 & z_1 & w_1 \\
        y_1^\ast & x_1^\ast & w_1^\ast & z_1^\ast \\
        x_2 & y_2 & z_2 & w_2 \\
        y_2^\ast & x_2^\ast & w_2^\ast & z_2^\ast
    \end{bmatrix}
\end{eqnarray}
where the coefficients are generally complex. The transformation is a \textit{Bogoliubov transformation} if the new operators $b_1$ and $b_2$ are again commuting boson modes. 

\begin{lemma}[Triviality]
We say that a Bogoliubov transformation is trivial if it does not mix the $a_1,a_1^\ast$ with the $a_2,a_2^\ast$. We note that if $b_1\equiv x_1a_1+y_1a_1^\ast$, that is, $z_1=w_1=0$, then $b_2 \equiv z_2 a_2 + w_2 a_2^\ast$.
\end{lemma}
\begin{proof}
    If $z_1=w_1=0$, then we obtain 
    $i) \, |x_1|^2 = |y_1|^2 +1, \, ii) \, x_1x_2^\ast = y_1 y_2^\ast $, and $ iii) \, x_1y_y=y_1 x_2$. Multiplying $ii)$ by $y_2$ and using $iii)$ yields $y_1 ( |x_2 |^2 - |y_2|^2 ) =0$. Likewise, multiplying $ii)$ by $x_2$ and using $iii)$ yields $x_1 ( |x_2 |^2 - |y_2|^2 ) =0$. As we cannot have both $x_1$ and $y_1$ vanishing, it must be the case that $ |x_2 |^2 - |y_2|^2 =0$.

    From $ii)$, we have that $|x_1|^2 |x_2|^2 = |y_1|^2 |y_2|^2$, however, since $ |x_2 |^2 = |y_2|^2 $, we either have that both $x_2$ and $y_2$ vanish or that $|x_1|^2=|y_1|^2$. But the latter is ruled out by $i)$. Therefore, we must have $x_2=y_2$. Similarly, setting $x_1=y_1=0$ would imply that $w_2=z_2=0$.
\end{proof}

\begin{proposition}
\label{prop_inverse}
    The transformation (\ref{eq:2mode}) is Bogoliubov provided that the coefficient matrix $S$ is invertible with
    \begin{eqnarray}
        S^{-1} =TS^\ast T, \quad \text{where}
        \, \,
        T= \begin{bmatrix}
            1 &0 & 0 & 0 \\
            0 &-1 & 0 & 0 \\
            0 &0 & 1 & 0 \\
            0 &0 & 0 & -1 \\
        \end{bmatrix}
    \end{eqnarray}
    and this implies
    \begin{eqnarray}
        x_ix_j^\ast +z_iz_j^\ast -y_i y_j^\ast -w_iw_j^\ast &=& \delta_{ij} \\
        x_1y_2-y_1 x_2+z_1w_2 -w_1 z_2&=&0.
    \end{eqnarray}
\end{proposition}
They follow from the requirements that $[b_i,b_j^\ast ]=\delta_{ij}$ and $[b_1,b_2]=0 $, respectively. The matrix form is adapted from, for example, \cite{GJN}.

\begin{corollary}
        The Bogoliubov transformation (\ref{eq:2mode}) can be inverted to obtain
    \begin{eqnarray}
        a_1 &=& x_1^\ast \, b_1 +x_2^\ast \, b_2 -y_1 \, b_1^\ast -y_2\, b_2^\ast \nonumber \\
        a_2 &=& z_1^\ast \, b_1 + z_2^\ast \, b_2 - w_1  \, b_1^\ast - w_2 \, b_2^\ast 
        \end{eqnarray}
\end{corollary}

If we take the joint vacuum state $|0 \rangle\otimes |0 \rangle$ for the modes $a_1$ and $a_2$, then $b_1$ and $b_2$ obtained through a Bogoliubov transformation will have a Gaussian state with the joint covariance matrix taking the block form
    \begin{eqnarray}
    C_{\text{two-mode}}=
        \left[ 
        \begin{array}{cccc}
             \langle b_1 b_1^\ast \rangle & \langle b_1 b_1 \rangle &\langle b_1 b_2^\ast \rangle & \langle b_1 b_2 \rangle  \\
             \langle b_1^\ast b_1^\ast \rangle & \langle b_1^\ast b_1\rangle &\langle b_1^\ast b_2^\ast \rangle & \langle b_1^\ast b_2 \rangle  \\
             \langle b_2 b_1^\ast \rangle & \langle b_2 b_1\rangle &\langle b_2 b_2^\ast \rangle & \langle b_2 b_2 \rangle  \\
             \langle b_2^\ast  b_1^\ast \rangle & \langle b_2^\ast b_1\rangle &\langle b_2^\ast b_2^\ast \rangle & \langle b_2^\ast b_2 \rangle  \\
        \end{array}
        \right]=
        \left[ 
        \begin{array}{cc}
            C_1 & K_{12} \\
             K_{21} & C_2
        \end{array}
        \right] 
        \label{eq:correlationmatrices}
    \end{eqnarray}
where 
\begin{eqnarray}
    C_i=
\left[ \begin{array}{cc}
     n_i+1 &  m_i\\
     m^\ast_i   &   n_i
\end{array}
\right], \quad
K_{12} =
\left[ \begin{array}{cc}
     v &  u\\
     u^\ast   &   v^\ast
\end{array}
\right] , \,
\quad
K_{21} =
\left[ \begin{array}{cc}
     v^\ast &  u\\
     u^\ast   &   v
\end{array}
\right]
.
\end{eqnarray}
The blocks $C_1$ and $C_2$ will be one-mode covariance matrices, while $K_{12}=K_{21}^\ast$ give the cross-correlations. We shall look at these in more detail next. First, we note that
\begin{eqnarray}
    C_{\text{two-mode}}= S \begin{bmatrix}
        C_{\text{vac}} & 0\\
        0 & C_{\text{vac}}
    \end{bmatrix}
    S^\ast.
    \label{eq:C_block}
\end{eqnarray}
\begin{proposition}[Marginal states]
    \label{prop_marg}
    The modes $b_1$ and $b_2$ are separately Gaussian with parameters $(n_1,m_1)$ and $(n_2,m_2) $, respectively,
    where
    \begin{eqnarray}
         y_iy_i^\ast +w_iw_i^\ast &=&n_i , \\
         x_iy_i +z_iw_i &=& m_i .
    \end{eqnarray}
    Furthermore, the Bogoliubov identities give the related identity 
\begin{eqnarray}
    x_ix_i^\ast +z_iz_i^\ast =n_i +1.     
    \label{eq:n+1}
\end{eqnarray}
\end{proposition}

The new modes are, however, correlated in general and this is described below.

\begin{proposition}[Cross-mode correlations]
    \label{lem4}
    The entries in the correlation matrices are
    \begin{eqnarray}
        x_1x_2^\ast +z_1z_2^\ast  \, (=y_1^\ast y_2 +w_1^\ast w_2 )= v,\\
        x_1 y_2 +z_1 w_2 \, (= y_1x_2+w_1z_2) =u .
    \end{eqnarray}
\end{proposition}

By Gaussianity, the joint state of the $b$ modes will factor if they are uncorrelated: $\langle b_1^\ast b_2 \rangle =0 $ and $\langle b_1 b_2 \rangle =0 $. In this case, we have the following:
\begin{eqnarray}
    \langle e^{i\sum_{k=1,2}(t_k^\ast b_k + t_k\, b_k^\ast )}\rangle
    =\langle e^{i(t_1^\ast b_1 + t_1\, b_1^\ast )}\rangle
    \, \langle e^{i(t _2^\ast b_2 + t_2\, b_2^\ast )}\rangle.
\end{eqnarray}

\begin{proposition}
\label{prop:2}
    If the induced quantum state on the modes obtained through a Bogoliubov transformation factorizes, then the squeezing is necessarily maximal.
\end{proposition}
\begin{proof}
    The relation (\ref{eq:C_block}) tells us that $C_{\text{two-mode}}$ has two (orthogonal) null eigenvectors given by
    \begin{eqnarray}
        e_1 = (S^\ast)^{-1} \begin{bmatrix}
            0 \\
            1\\
            0\\
            0
        \end{bmatrix}
        =\begin{bmatrix}
            -y_1 \\
            x_1^\ast \\
            -y_2 \\
            x_2^\ast
        \end{bmatrix}, \quad
         e_2 = (S^\ast)^{-1} \begin{bmatrix}
            0 \\
            0\\
            0\\
            1
        \end{bmatrix}
        =\begin{bmatrix}
            -w_1 \\
            z_1^\ast \\
            -w_2 \\
            z_2^\ast
        \end{bmatrix}.
    \end{eqnarray}
    We have, for instance, 
    \begin{eqnarray}
        0= e_1^\ast\,  C_{\text{two-mode}} \, e_1= e_1^\ast\, \begin{bmatrix}
            C_1 & 0 \\
            0 & C_2
        \end{bmatrix} \, e_1
        = f_{11}^\ast \, C_1 f_{11} + f_{12}^\ast C_2 f_{12}
    \end{eqnarray}
    where $f_{11} = \begin{bmatrix}
            -y_1 \\
            x_1^\ast \\
        \end{bmatrix}$
    and $f_{12} = \begin{bmatrix}
            -y_2 \\
            x_2^\ast
        \end{bmatrix}$.

        From the positivity of $C_1$ and $C_2$, we have that both $f_{11}^\ast \, C_1 f_{11}=0$ and $ f_{12}^\ast C_2 f_{12}=0$. Either $f_{11}=0$ and so the Bogoliubov transformation is trivial, or $\text{det}(C_1)=0$. However, the trivial case automatically implies maximality.
\end{proof}

\bigskip

This result tells us that correlations automatically occur in the sub-maximal case.

\subsection{Balanced Bogoliubov Transformations}

The Araki-Woods representation is the well-known construction for a thermal state given by
\begin{eqnarray}
    b_1 &=& \sqrt{n+1}\, a_1 + \sqrt{n}\, a_2^\ast , \nonumber \\
    b_2 &=& \sqrt{n}\, a_1^\ast + \sqrt{n+1} \, a_2 .
\end{eqnarray}
Here, the modes have thermal state marginals with number parameter $n$. The first mode gives the standard Gel'fand-Naimark-Segal (GNS) construction, while the second mode is the commuting representation coming from Tomita-Takesaki theory. In this case, we see that
\begin{eqnarray}
    C= \left[ 
    \begin{array}{cc}
         n+1&0  \\
         0& n 
    \end{array}
    \right] 
    ,
    \qquad
    K= \left[ 
    \begin{array}{cc}
         0 & \sqrt{n(n+1)}  \\
         \sqrt{n(n+1)} & 0
    \end{array}
    \right] .
    \end{eqnarray}
We note that even though the modes $b_1$ and $b_2$ yield commuting representation of the thermal state, with the joint vacuum of $a_1$ and $a_2$ as GNS state, they are nevertheless correlated!

It is natural to ask if there is a similar construction leading to two modes which have identical squeezed marginals. To this end, we observe that the Bogoliubov and mode state coefficient identities are symmetric under the combined interchanges $x\leftrightarrow z$ \textit{and} $y \leftrightarrow w$. This suggests the following class of Bogoliubov transformations which we term balanced (the Araki-Woods case being an example). We shall often just say that the two modes obtained by this way are balanced.

\begin{definition}
\label{def:BBT}
    We say that a Bogoliubov transformation of the form (\ref{eq:2mode}) is \textbf{balanced} if the coefficients of the two modes are related by  
    \begin{eqnarray}
        x_2=z_1, \quad y_2 = w_1, \quad z_2 = x_1 , \quad w_2 = y_1 .
        \label{eq:interchange}
    \end{eqnarray}
    A balanced Bogoliubov transformation is therefore of the form
    \begin{eqnarray}
    b_1 &=& x \, a_1+ z \, a_2 + y  \, a_1^\ast  + w \, a_2^\ast , \nonumber \\
    b_2 &=& z \, a_1 + x \, a_2+ w \, a_1^\ast  + y \, a_2^\ast .
    \label{eq:BalBog}
\end{eqnarray}
\end{definition}

From (\ref{eq:2mode}), we see that the parameters $x,y,z,w$ are complex and are required to satisfy the identities
\begin{eqnarray}
    |x|^2+|z|^2-|y|^2- |w|^2 =1, \\
    x z^\ast + z x^\ast - y w^\ast -w y^\ast =0.
    \label{eq:Bod_ids}
\end{eqnarray}

The induced states of $b_1$ and $b_2$ have the same number and squeezing parameters. In fact, when $C$ is given by (\ref{eq:C}) and the inter-mode correlations are given by
    \begin{eqnarray}
        K
        =
        \left[ 
        \begin{array}{cc}
            u & v \\
            v^\ast & u
        \end{array}
        \right] 
    \end{eqnarray}
    with state parameters being
    \begin{eqnarray}
        n= |y|^2+|w|^2, \qquad m = xy+zw , \nonumber \\
        v= 2 \, \text{Re} (y^\ast w)\equiv 2 \, \text{Re} (x^\ast z), \qquad 
        u = x w + z y.
        \label{eq:uv}
    \end{eqnarray}

\textbf{Example (Bishop-Vourdas):} An example of a balanced Bogoliubov transformation is obtained as follows. We select the following unitaries,
\begin{eqnarray}
    R_i(\phi)=e^{i\phi a_i^\ast a_i}, \quad
    U_i(\rho ) = e^{- \frac{\rho}{4}(a_i^2-a_i^{\ast 2})}, \quad
    V(r) = e^{-\frac{1}{2}(a_1a_2-a_1^\ast a_2^\ast)}
\end{eqnarray}
and set $U=R_1(\theta/2)R_2(\theta /2) U_1 (\rho) U_2 (\rho ) V(r)$. Then $b_i = U^\ast a_iU$ is a balanced Bogoliubov transformation \cite{Bishop} with
\begin{eqnarray}
    x= e^{i \theta /2} c_rc_\rho,\quad
    y= e^{i \theta /2} c_rs_\rho,\quad 
    z= e^{i \theta /2} s_r s_\rho,\quad
    w= e^{i \theta /2} s_r c_\rho,
    \label{eq:BV}
\end{eqnarray}
where $c_r= \cosh (r/2), s_r=\sinh (r/2), \, c_\rho = \cosh (\rho/2), s_\rho = \sinh (\rho/2)$. In this case, $n=c_r^2s_\rho^2+s_r^2c_\rho^2$, $m=e^{i\theta} (c_r^2+s_r^2)c_\rho s_\rho$. (One may check that $\Delta = \frac{1}{4}\sinh^2 r$ in this case.)

\bigskip

\textbf{Example (Hellmich-Honegger-K\"{o}stler-K\"{u}mmerer-Rieckers):} An alternative one is the following one \cite{HKKR}:
\begin{eqnarray}
    x =\kappa_+ (\frac{\rho}{2}+\frac{m}{2}+\frac{1}{4}), \quad
    y  =\kappa_+ (\frac{\rho}{2}+\frac{m}{2}-\frac{1}{4}), \nonumber \\
    z =\kappa_- (\frac{\rho}{2}-\frac{m}{2}-\frac{1}{4}), \quad
    w =\kappa_- (\frac{\rho}{2}-\frac{m}{2}+\frac{1}{4}) .    
    \label{eq:xyzw_H2K2R}
\end{eqnarray}
where $\rho = \sqrt{|m|^2 + \frac{1}{4}}$ and $\kappa_\pm=\sqrt{\frac{n+ \frac{1}{2} \pm \rho }{\rho (\rho \pm \mathrm{Re} \, m)}}$. These lead to a balanced pair of modes $b_1$ and $b_2$ which separately have a Gaussian state with the same $n$ and $m$ (assumed sub-maximal).

\bigskip

We note the general feature that the two modes are correlated ($K_{12 } \neq 0$) when the squeezing is sub-maximal.

Balanced Bogoliubov transformations (\ref{eq:BalBog}) are readily inverted as follows: 
    \begin{eqnarray}
        a_1 &=& x ^\ast b_1 + z^\ast b_2 -y\, b_1^\ast - w \, b_2 ^\ast ,\nonumber \\
        a_2 &=& z ^\ast b_1 + x^\ast b_2- w\, b_1^\ast - y \, b_2^\ast .
    \end{eqnarray}


\section{Symplectic Structure on Hilbert Spaces}
\label{sec:symp}
This section recalls content which is fairly standard and we follow several sources \cite{HL85,Belton,Bhat,HKKR}
Let $\mathfrak{h}$ be a separable complex Hilbert space and let us decompose its inner product into
\begin{eqnarray}
    \langle \phi | \psi \rangle = g(\phi , \psi ) + i \, \sigma (\phi , \psi )
\end{eqnarray}
where
\begin{eqnarray}
    \text{(metric form)}&& \quad g (\phi , \psi )= \text{Re}\langle \phi |\psi \rangle , \nonumber \\
    \text{(symplectic form)}&& \quad \sigma (\phi , \psi )= \text{Im}\langle \phi |\psi \rangle .
\end{eqnarray}
We note that $g$ and $\sigma$ are symmetric and anti-symmetric forms, respectively, and that both are non-degenerate. Moreover, they are only \textit{real} linear in their arguments though are related by
\begin{eqnarray}
    g(\phi , \psi )= \sigma (\phi , i \psi ) , \qquad
    \sigma (\phi , \psi) = g( i \phi , \psi) .
    \label{eq:g_sigma}
\end{eqnarray}

For a given set of vectors $D\subset \mathfrak{h}$, its symplectic complement is defined to be
\begin{eqnarray}
    D'\triangleq \{ \phi \in \mathfrak{h}:
    \sigma (\phi, \psi ) =0, \, \forall \psi \in D \}.
\end{eqnarray}
From the relation (\ref{eq:g_sigma}) between the metric and symplectic forms, we see that the symplectic complement $D'$ may alternatively be described as the $g$-orthogonal complement of $i D$.

In general, a real Hilbert subspace is said to be \textit{standard} if
\begin{eqnarray*}
    (i) &&\mathfrak{k}+i\mathfrak{k} \text{ is dense in }\mathfrak{k }, \text{ and}\\
    (ii) && \mathfrak{k}\cap i \mathfrak{k}\equiv \{0 \}.
\end{eqnarray*}
One may readily show that a real Hilbert subspace is standard if and only if its symplectic complement is standard.

\subsection{Complex Conjugations}
A complex conjugation $\jmath$ on $\mathfrak{h}$ is a bijection which is anti-unitary (i.e, $\langle \jmath \phi | \jmath \psi \rangle = \langle \psi | \phi \rangle$ for all $\phi,\psi \in \mathfrak{h}$) and involutive ($\jmath^2 = I$). (The familiar complex conjugation operation for standard representations of Hilbert spaces being the simplest example.) It follows that $\jmath^{-1}\equiv \jmath^\ast$. Complex conjugations are easily seen to be real linear, but not complex linear; indeed, $\jmath c\psi = c^\ast \jmath \psi$ for complex scalars $c$.

The conjugation leaves the metric invariant, that is, $g(\jmath \phi , \jmath \psi ) =g(\phi , \psi )$. It similarly introduces a sign change for the symplectic form.

For each vector $\psi \in \mathfrak{h}$, we introduce the \lq\lq quadrature\rq\rq\, vectors
\begin{eqnarray}
    \psi_q \triangleq \frac{1}{2}\psi  + \frac{1}{2} \jmath \psi  , \qquad
    \psi_p \triangleq \frac{1}{2}\psi  - \frac{1}{2} \jmath \psi 
\end{eqnarray}
and these have the symmetries $\jmath \psi_q= \psi_q$ and $\jmath \psi_p = - \psi_p$. Let us write $\mathfrak{h}_q$ and $\mathfrak{h}_p$ for the closure of $\{\psi_q : \psi \in \mathfrak{h}\}$ and $\{ \psi_p : \psi\in \mathfrak{h}\}$ respectively. These are \textit{real} Hilbert spaces which we refer to as the quadrature spaces corresponding to the complex conjugation $\jmath$. As every vector admits the decomposition $\psi = \psi_q+\psi_p$, we have $\mathfrak{h}=\mathfrak{h}_q+\mathfrak{h}_p$. We also see that the only vector common to both quadrature spaces is the zero vector. Therefore, $\mathfrak{h}_q$ is standard, as is $\mathfrak{h}_p$.

As $\jmath i \psi =-i \jmath \psi$, we see that $\mathfrak{h}_q \equiv i \mathfrak{h}_p$ and $\mathfrak{h}_p \equiv i \mathfrak{h}_q$. In fact, the conjugation is determined by specifying its $q$-quadrature space: equivalently specifying a complete basis. We may identify $\mathfrak{h}$ as either $\mathfrak{h}_q+ \mathfrak{h}_p$, or as $\mathfrak{h}_q + i \mathfrak{h}_q$. The subspaces are not generally orthogonal, but are $g$-orthogonal: let $\phi \in\mathfrak{h}_q$ and $\psi \in\mathfrak{h}_p$ then $g(\phi,\psi)=g(\jmath \phi, \jmath \psi)=-g(\phi,\psi)$.

Let $A$ be a (complex) linear operator on $\mathfrak{h}$ with adjoint $A^\ast$, then we define its conjugation and transpose, respectively, by
\begin{eqnarray}
    A^\# = \jmath A\jmath, \qquad A^\top = \jmath A^\ast \jmath .
\end{eqnarray}

It is convenient to introduce the \lq\lq complexified\rq\rq\, Hilbert space $\mathfrak{h}\oplus\jmath \mathfrak{h}$. We take $\breve{\mathfrak{h}}$ as the closed real subspace spanned by vectors of the form
\begin{eqnarray}
    \Breve{\psi}
    = \left[
    \begin{array}{c}
        \psi  \\
         \jmath \psi 
    \end{array}
    \right], \qquad \psi \in \mathfrak{h}.
\end{eqnarray}

\subsection{Symplectomorphisms}

The linear symplectomorphisms form the set $Sp(\mathfrak{h})$ the real linear bijections $S:\mathfrak{h}\mapsto\mathfrak{h} $ leaving the symplectic form invariant, that is, $\sigma (S,\phi,S\psi)= \sigma (\phi , \psi)$.

Every linear symplectomorphism can be written as $S\phi = S_- \phi + S_+ \jmath \phi$ where $S_\pm$ are (complex) linear operators on $\mathfrak{h}$.

The symplectomorphism $S$ then naturally induces an operator  $\breve{S}$ on $ \mathfrak{h}\oplus\jmath \mathfrak{h}$ given by
\begin{eqnarray}
    \breve{S} =
    \left[  
    \begin{array}{cc}
         S_-& S_+  \\
         S_+^\#& S_-^\#
    \end{array}
    \right] ,
\end{eqnarray}
so that $\phi = S\psi$ implies $\breve{\phi}= \breve{S} \breve{\psi}$.

The condition that $S$ be a symplectomorphism may be equivalently stated as $\breve{S}$ being invertible on $ \mathfrak{h}\oplus\jmath \mathfrak{h}$ with inverse
\begin{eqnarray}
    \breve{S}^{-1} =
    \left[  
    \begin{array}{rr}
         S_-^\ast& - S_+^\top  \\
         -S_+^\ast& S_-^\top
    \end{array}
    \right] .
\end{eqnarray}

\subsection{Second Quantization}
The CCR algebra over a separable Hilbert space $\mathfrak{h}$ is the C*-algebra generated by elements $w(\phi)$, for $\phi \in \mathfrak{h}$, with the relations $w(\phi) w (\psi) = e^{-i \sigma (\phi , \psi )} w(\phi+\psi )$ and $w(\psi)^\ast = w (-\psi)$. By a result of Slawny, the CCR algebra is unique up to isomorphism and will be denoted as $CCR(\mathfrak{h})$.

The algebra has the Weyl-Fock representation $\pi_{\text{Fock}}$ on the Fock space $\Gamma (\mathfrak{h})$ with $w(\psi )$ represented by the Weyl-Fock operator $W(\psi)$. The Fock vacuum $|\text{vac}\rangle$ is cyclic for this representation and induces the state $\langle \text{vac} | W(\phi ) \, \text{vac} \rangle =e^{- \frac{1}{2} \| \phi \|^2 }.$ We note that the Weyl-Fock operator
\begin{eqnarray}
    W(\phi) = e^{a(\phi) - a^\ast (\phi )}
\end{eqnarray}
where $a^\ast (\cdot)$ and $a(\cdot)$ are the (complex-linear) creation operator and the (complex-anti-linear) annihilation operator on $\Gamma (\mathfrak{h})$, respectively.

If $\mathfrak{k}$ is a real subspace of $\mathfrak{h}$ then we may similarly define $CCR (\mathfrak{k})$ which we identify with its representation as a subalgebra of the Weyl-Fock representation. The corresponding von Neumann algebra is
\begin{eqnarray}
    \mathfrak{M}(\mathfrak{k})= \pi_{\text{Fock}} (CCR (\mathfrak{k}))''.
\end{eqnarray}
(Here $\mathfrak{M}'$ denotes the commutant of an algebra with the bounded operators on $\mathfrak{h}$.) 

A key result is the \textit{Duality Theorem} which states that, for closed standard subspaces $\mathfrak{k}$, 
\begin{eqnarray}
    \mathfrak{M} (\mathfrak{k})'= \mathfrak{M}(\mathfrak{k}') 
\end{eqnarray}
and that, moreover, the Fock vacuum is cyclic and separating for both $\mathfrak{M} (\mathfrak{k})$ and $\mathfrak{M} (\mathfrak{k})'$.
In particular, condition \textit{(i)} implies that $\mathfrak{M} (\mathfrak{k}')\subset \mathfrak{M} (\mathfrak{k})'$ while condition \textit{(ii)} ensures equality.

\subsection{Representations Inducing Quasi-Free States}
The general (multi-mode) case is outlined in \cite{Belton} and utilizes the general form for Bogoliubov transformations derived in \cite{HR}. However, for our purposes, it suffices to start with a simpler two-mode squeezing problem and fix a balanced set of complex parameters $(x,y,z,w)$ as in Definition \ref{def:BBT}. In analogy with the two-mode squeezing in (\ref{eq:BalBog}, we set
\begin{eqnarray}
    b(\phi ) &=& x\, a (\phi)\otimes I + z\, I\otimes a(\phi) +y\, a^\ast(\jmath \phi) \otimes I +w\, I \otimes a^\ast (\jmath \phi ) , \nonumber \\
    b'(\phi ) &=& z\, a (\phi)\otimes I + x\, I\otimes a(\phi) +w\, a^\ast(\jmath \phi) \otimes I +y\, I \otimes a^\ast (\jmath \phi ).
\end{eqnarray}
Note that we need to include the conjugation $\jmath$ to ensure that $b(\cdot )$ is overall anti-linear. The new operator should be thought of as acting on the Hilbert space
\begin{eqnarray}
    \Gamma (\mathfrak{h}\oplus \jmath \mathfrak{h}) \cong \Gamma (\mathfrak{h} )\otimes \Gamma ( \jmath \mathfrak{h}).
\end{eqnarray}
We see that 
\begin{eqnarray}
    b(\phi ) - b^\ast (\phi ) = a(x^\ast \, \phi - y \jmath  \phi ) \otimes I
    + I \otimes a^\ast ( z^\ast \, \phi -w\, \jmath\phi ) - \text{H.c.}
\end{eqnarray}
and it follows that
\begin{eqnarray}
    \pi \left( e^{b(\phi ) - b^\ast (\phi )} \right) \equiv W(x^\ast \, \phi - y \jmath  \phi ) \otimes
    W ( z^\ast \, \phi -w\, \jmath\phi ) .
\end{eqnarray}
gives a representation of the CCR for which the joint Fock vacuum $|\text{vac}\rangle \otimes | \text{vac} \rangle$ is cyclic and separating and induces a Gaussian state.

To see this in more detail, let us introduce the following operators (understood as acting on $\mathfrak{h}\oplus \jmath \mathfrak{h}$)
\begin{eqnarray}
    \Sigma 
     = \left[ 
    \begin{array}{cc}
        x^\ast & -y \\
        z^\ast & -w
    \end{array}\right] , 
    \qquad
    \Sigma^\prime  =
    \left[ 
    \begin{array}{cc}
        z^\ast & -w \\
        x^\ast & -y
    \end{array}\right] .
\end{eqnarray}
By the symmetries inherent to balanced Bogoliubov transformations, we see that $(x^\ast,-y,-z,w^\ast )$ again yields the parameters of a Bogoliubov transformation. (That is, the requirements \ref{eq:Bod_ids} are invariant under the replacement $x\to x^\ast , z \to z^\ast , y \to -y, w \to -w$.)

Associated with these are the real subspaces $\mathfrak{k}$ and $\mathfrak{k}'$ obtained as the closure of the sets $\{ \Sigma \check \psi: \psi \in \mathfrak{h} \}$ and $\{ \Sigma' \check \psi: \psi \in \mathfrak{h} \}$, respectively. These subspaces are symplectic complements to each other.


\subsection{Non-Fock Quantum Stochastic Calculus}
\label{sec:nonfock}
The Fock space $\Gamma (\mathfrak{h})$ over a given one-particle space $\mathfrak{h}$ is the direct sum of symmetrized $n$-particle spaces $\otimes_{\text{symm.}}^n \mathfrak{h}$. The traditional bosonic creation and annihilation operators with test-function $f\in \mathfrak{h}$ are denoted as $A(f)$ and $A^\ast(f)$, respectively.

The Hudson-Parthasarathy quantum stochastic calculus uses the one-particle space $L^2 (\mathbb{R}_+,dt)$ one defines the annihilation process to be the operators $A_t \triangleq A(1_{[0,t]})$ for $t\ge 0$. We therefore have the canonical commutation relations
\begin{eqnarray}
    [ A_t , A_s^\ast ]= \min (t,s) .
\end{eqnarray}

The increments are understood in the Ito sense as future-pointing: that is, informally $dA_t = A_{t+dt} -A_t$ for $dt>0$, etc. Central to the calculus is the quantum Ito table:
\begin{eqnarray}
        \centering
        \begin{tabular}{c|cc}
           $ \times$ & $dA$ & $dA^\ast$  \\
            \hline
             $dA$   &  0 & $dt$ \\
             $dA^\ast$ &0 &0
        \end{tabular}
        .
\end{eqnarray}

To obtain non-Fock noise, we fix a Bogoliubov matrix $S\in Sp(2)$ and on the double Fock space $\Gamma (L^2(\mathbb{R}_+,dt) \otimes \Gamma (L^2(\mathbb{R}_+,dt)$ set
\begin{eqnarray}
    B_t &=& x_1 \, A_t \otimes I + y_1 \, A_t^\ast \otimes I + z_1 \, I \otimes A_t + w_1 \, I \otimes A_t^\ast , \nonumber \\
    B'_t  &=& x_2 \, A_1 \otimes I + y_2 \, A_t^\ast \otimes I + z_2 \, I \otimes A_t + w_2 \, I \otimes A_t^\ast ,
    \label{eq:2mode_noise}
\end{eqnarray}
where we use the entries from (\ref{eq:2mode}). The processes $B_t , B^\ast_t$ commute with the $B_t', B_t^{\prime\ast}$ and together they satisfy the quantum Ito table
\begin{eqnarray}
            \centering
        \begin{tabular}{c|cc}
           $ \times$ & $dB^\ast \quad dB$ & $dB^{\prime\ast} \quad dB'$  \\
            \hline
             $\begin{array}{c}
                  dB \\
                  dB^\ast
             \end{array}$   &  $C_1 \, dt$ & $K_{12} \, dt$ \\
             $\begin{array}{c}
                  dB' \\
                  dB^{\prime\ast}
             \end{array}$ & $K_{21} \, dt $ & $C_2 \, dt$
        \end{tabular}
    \end{eqnarray}
with the matrix blocks given by \ref{eq:correlationmatrices}). We shall write $(n,m)$ for the parameters $(n_1,m_1)$ of the first noise and $(n',m')$ for the parameters $(n_2,m_2)$ of the second.

\bigskip

A unitary adapted process $U_t$ can now be obtained from a quantum stochastic differential equation of the form
\begin{eqnarray}
    dU_t = \bigg\{ L\otimes dB_t^\ast - L^\ast \otimes dB_t
    +K \otimes dt \bigg\} U_t ,
    \label{eq:U_QSDE}
\end{eqnarray}
where $dt$-term takes the form
\begin{eqnarray}
    K=- \frac{1}{2} (n+1) L^\ast L - \frac{1}{2} n LL^\ast 
    -\frac{1}{2} m L^{\ast 2} - \frac{1}{2}m^\ast L^2
    - i H.
    \label{eq:K}
\end{eqnarray}

For $X$ an observable on $\mathfrak{h}_0$, we define its evolution under the unitary quantum stochastic dynamics as
\begin{eqnarray}
    j_t (X) = U_t^\ast (X \otimes I) U_t.
\end{eqnarray}
From the quantum Ito calculus, on has
\begin{eqnarray}
    dj_t (X) = j_t ( [X,L]) \otimes dB^\ast_t
    +j_t ([L^\ast ,X] \otimes dB_t
    +j_t ( \mathcal{L} X ) \otimes dt
\end{eqnarray}
where we encounter the Lindblad generator 
\begin{eqnarray}
    \mathcal{L} X &=& (n+1)
    \left\{\frac{1}{2}  [L^\ast ,X]L + \frac{1}{2} L^\ast [X,L ]\right\} 
    + n\left\{\frac{1}{2}  [L ,X]L^\ast + \frac{1}{2} L [X,L^\ast ]\right\} \nonumber \\
    &&+m\left\{\frac{1}{2}  [L^\ast ,X]L^\ast + \frac{1}{2} L^\ast [X,L^\ast ]\right\}
    +m^\ast \left\{\frac{1}{2}  [L ,X]L + \frac{1}{2} L [X,L  ]\right\} \nonumber \\
    &&   -i[X,H] .
    \label{eq:L}
\end{eqnarray}

\begin{proposition}
The output process is defined as $B_t^{\mathrm{out}}= U^\ast_t (I \otimes B_t ) U_t$ and we have
\begin{eqnarray}
    dB_t^{\mathrm{out}} = dB_t + j_t (L) dt .
    \label{eq:io}
\end{eqnarray}
\end{proposition}
\begin{proof}
    While (\ref{eq:io}) is formally identical to the $n=0$ case, the derivation deserves some attention. Retaining the non-trivial quantum Ito corrections, we find
    \begin{eqnarray}
        dB_t^{\text{out}} &=& dB_t + (dU_t^\ast) (I \otimes dB_t) U_t + U_t^\ast (I \otimes dB_t) dU_t
        \nonumber \\
        &=& dB_t + U_t^\ast \big( (n+1)L -n\, L +m\, L^\ast -m\, L^\ast \big) dt, 
    \end{eqnarray}
    so that $n,m$-dependent terms cancel leaving the Fock vacuum expression.
\end{proof}


\section{Quantum Filtering for Squeezed Noise}
\label{sec:filter}
Without loss of generality, we take the state of the system to be the pure state corresponding to a vector $\phi$ and with an abuse of notation just write $\mathbb{E}$ for the separable state where the system is in state $\phi $ and the bath is in the squeezed state. The latter then corresponds to the joint Fock vacuum $\Psi = \Phi \otimes \Phi$ for the $A$-noises.

We shall consider measurements of the output quadrature $Y_t = B_t^{\text{out}} +B_t^{\text{out}\ast} $. The von Neumann algebra $\mathfrak{Y}_t$ generated by the observables $\{ Y_s: s\in [0,t]\}$ is commutative. Our aim is to calculate the conditional expectation $\pi_t (X)$ of an observable $j_t (X)$ onto the measurement algebra $\mathfrak{Y}_t$:
\begin{eqnarray}
    \pi_t (X) \triangleq \mathbb{E} [j_t (X)|\mathfrak{Y}_t].
\end{eqnarray}
This conditional expectation is guaranteed to exist since we have the nondemolition property $j_t(X) \in \mathfrak{Y}_t'$ and so one may use the standard spectral theorem construction from \cite{BvH}.

The output quadrature may be written as $Y_t = U_t^\ast  ( I \otimes Z_t) U_t$ where $Z_t = (B_t+B_t^\ast)$. The von Neumann algebra generated by the observables $\{ Z_s: s\in [0,t]\}$ is likewise commutative and is denoted $\mathfrak{Z}_t$. We note that 
\begin{eqnarray}
    dY_t = dZ_t + j_t (L+L^\ast) \, dt
\end{eqnarray}
and that
\begin{eqnarray}
    dY_t \, dY_t = dZ_t \, dZ_t =(2n+1+2\, \text{Re}\,m) \, dt.
    \label{eq:dY}
\end{eqnarray}

\subsection{Fixing A Commuting Quadrature}
We now come to a critical part of the construction. For a given phase $\lambda\in [0 ,2\pi )$, we may define a quadrature of the commutant field by $Z_t'=e^{i\lambda }B_t'+e^{-i\lambda} B_t^{\prime \ast}$.

The processes $Z$ and $Z'$ are self-adjoint commuting processes and therefore are essentially classical and may be embedded in the same Kolmogorov model. Moreover, for the choice of the joint vacuum state, they are jointly Gaussian. Let us suppose that we work with a balanced Bogoliubov representation, then we see that $Z'$ is mean zero with $dZ'_t \, dZ_t' = \big( 2n+1+2\, \text{Re}\, (e^{i2\lambda}m) \big) \, dt$. However, we also see that
\begin{eqnarray}
    dZ_tdZ_t'=dZ_t'dZ_t = 2 \bigg( \cos \lambda \,\big(v+ \text{Re}u \big)-\sin\lambda \, \text{Im}u \bigg) \, dt ,
\end{eqnarray}
where $u,v$ are the state parameters appearing in (\ref{eq:uv}).

\bigskip

It is clear that the choice $\lambda = \tan^{-1} \left( \frac{v+ \text{Re}u }{\text{Im}u} \right)$, modulo $\pi$, ensures that $Z$ and $Z'$ are uncorrelated and, moreover, statistically independent. In the special case where $\text{Im} \, u =0$, we understand that $\lambda \equiv \pi/2$.

\begin{definition}[Standing Assumptions]
    We shall say that the standing assumptions are met whenever i) the underlying Bogoliubov transformation behind the representation is balanced and ii) the phase parameter $\lambda$ has been chosen so that $Z$ and $Z'$ are statistically independent.
\end{definition}

For instance, for the coefficients in (\ref{eq:BV}) we should take $\lambda = \tan ^{-1} (\frac{2 \tau + \cos \theta}{\sin \theta})$ where $ \frac{1}{2} < \tau \le 1$
is given by $\tau = \frac{c_\rho^2}{c_\rho^2 +s_\rho^2}$. The choices for $\lambda$ are plotted in Figure \ref{fig:phase_Lambda}.

In the thermal case $(x=\sqrt{n+1},y=z=0,w=\sqrt{n})$, we find that $u=\sqrt{(n+1)n}$ and $v=0$. In this case, the choice is $\lambda = \pm \pi /2$.

\begin{remark}
    In the following, we are going to consider a measurement of both of the output quadratures corresponding to $Z$ and $Z'$. We aim to average out $Z'$ and this is why we took it to be statistically independent of $Z$. In physicist's language, we are fixing a representation and then choosing a suitable \lq\lq unraveling\rq\rq\, of the associated master equation - one corresponding to two independent quadratures. We have fixed $Z$, and so it was then a question of finding an appropriate $Z'$.

 It should be stressed that the only actual measurement that physically takes place is that of $Y$, the output quadrature corresponding to $Z$. However, unraveling supposes that we also measure the output corresponding to $Z'$ (which will be compatible) and average it out through the process of conditional expectation. 

    An alternative strategy is to fix $Z'$ as a specific quadrature (i.e., fix $\lambda$) and then select an appropriate representation (i.e, choose the parameters $(x,y,z,w)$ so that the quadratures are uncorrelated).
\end{remark}

\begin{figure}
    \centering
    \includegraphics[width=0.75\linewidth]{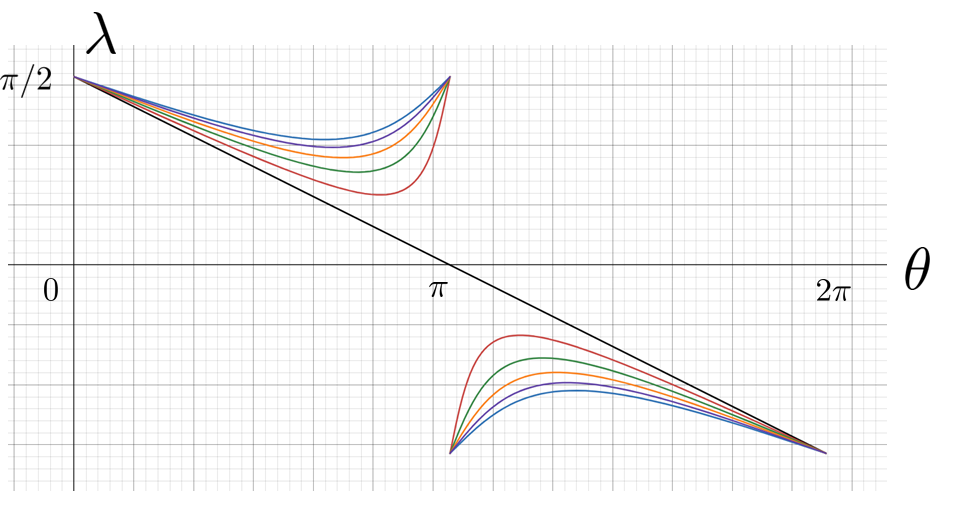}
    \caption{(color online) A plot of $\lambda$ versus $\theta$ for values of $\tau =0.5$ (black), $0.6$ (red), $0.7$ (green), $0.8$ (orange), 0.9 (purple), and 1.0 (blue).}
    \label{fig:phase_Lambda}
\end{figure}

\subsection{The Quantum Kallianpur-Striebel Formula}
The annihilators $A_t\otimes I$ and $I\otimes A_t$ kill the joint vacuum state $\Psi= \Phi \otimes \Phi$, but this will not be the case for the process $B_t$. In fact, we find
\begin{eqnarray}
\left[ 
    \begin{array}{c}
         dB  \\
        dB^\ast
    \end{array}
    \right] \Psi =
    \left[ 
    \begin{array}{cc}
        y_1 & w_1  \\
        x_1^\ast & z_1^\ast
    \end{array}
    \right] 
    \left[ 
    \begin{array}{l}
         dA^\ast \otimes I  \\
        I \otimes dA^\ast
    \end{array}
    \right] \Psi .
\end{eqnarray}
Similarly, 
\begin{eqnarray}
    \left[ 
    \begin{array}{l}
         dZ  \\
        dZ^\prime
    \end{array}
    \right] \Psi =
    \left[ 
    \begin{array}{cc}
       x_1^\ast + y_1 & z_1^\ast + w_1  \\
        x_2^\ast +y_2 & z_2^\ast +w_2
    \end{array}
    \right] 
    \left[ 
    \begin{array}{c}
         dA^\ast \otimes I  \\
        I \otimes dA^\ast
    \end{array}
    \right] \Psi .
\end{eqnarray}
As a consequence, we find that
\begin{eqnarray}
    \left[ 
    \begin{array}{c}
         dB  \\
        dB^\ast
    \end{array}
    \right] \Psi =
    \left[ 
    \begin{array}{cc}
        \alpha & \beta  \\
        \gamma & \delta 
    \end{array}
    \right]
    \, 
     \left[ 
    \begin{array}{l}
         dZ  \\
        dZ^\prime
    \end{array}
    \right] \Psi .
\end{eqnarray}

\begin{proposition}
\label{prop:alpha}
    In the case of a balanced representation, we have $\alpha + \gamma =1$ and $\beta = - \delta$ for the arbitrary choice of phase $\lambda$. We, moreover, have the identities
    \begin{eqnarray}
        (2n+1 + 2\text{Re}\, m ) |\alpha |^2 +  (2n+1 + 2\text{Re}\, (e^{2i \lambda }m)  ) |\beta |^2 &=& n ,
        \label{eq:id_n}\\
        (2n+1 + 2\text{Re}\, m ) |\gamma |^2 +  (2n+1 + 2\text{Re}\, (e^{2i \lambda }m)  ) |\delta |^2 &=& n+1 ,
        \label{eq:id_n+1}\\
        (2n+1 + 2\text{Re}\, m ) \gamma^\ast\alpha  +  (2n+1 +2 \text{Re}\, (e^{2i \lambda }m)  ) \delta^\ast \beta  &=& m ,
        \label{eq:id_m}
    \end{eqnarray}
    along with the explicit expressions
    \begin{eqnarray}
        \alpha = \frac{n+  m}{2n+1 +2 \,\text{Re} \, m}, \quad \gamma = \frac{n+1+ m^\ast }{2n+1 +2\, \text{Re} \, m}
        .
        \label{eq:alpha+gamma}
    \end{eqnarray}
\end{proposition}
\begin{proof}
The transfer matrix from the increments $dB_t\, \Psi$ and $dB^\ast_t \, \Psi$ to $dZ_t \, \Psi$ and $dZ'_t \, \Psi$ is
\begin{eqnarray}
    \left[ 
    \begin{array}{cc}
        \alpha & \beta  \\
        \gamma & \delta 
    \end{array}
    \right]
    =
    \left[ 
    \begin{array}{cc}
        y & w  \\
        x^\ast & z^\ast
    \end{array}
    \right]
    \,
    \left[ 
    \begin{array}{cc}
       x^\ast + y & z^\ast + w \\
        e^{i\lambda}w +e^{-i\lambda} z^\ast & e^{i\lambda}y +e^{-i\lambda}x^\ast
    \end{array}
    \right]^{-1}
    .
\end{eqnarray}
The determinant of the matrix to be inverted is
\begin{eqnarray}
    D= e^{i\lambda}(y^2+yx^\ast-w^2-wz^\ast ) +e^{-i\lambda}(x^{\ast 2}+yx^\ast -wz^\ast -z^{\ast 2}),
\end{eqnarray}
and, by elementary matrix inversion, we readily see that
\begin{eqnarray}
    \alpha &=&\frac{1}{D} \bigg( e^{i\lambda}(y^2 -w^2) +e^{-i\lambda}(yx^\ast -wz^\ast ) \bigg), \nonumber\\
    \gamma &=&\frac{1}{D} \bigg( e^{i\lambda}(yx^\ast-wz^\ast ) +e^{-i\lambda}(x^{\ast 2}-z^{\ast 2}) \bigg), \nonumber \\
    \delta &=&- \beta = \frac{1}{D}(wx^\ast - y z^\ast ) .
\end{eqnarray}
The identities $\alpha + \gamma =1$ and $\beta + \delta =0$ from the statement of the proposition follow immediately.

Next note that $\langle dB\,  \Psi ,dB\, \Psi \rangle = \langle \Psi ,  dB^\ast dB\, \Psi \rangle = n \, dt $, and substituting $dB= \alpha dZ + \beta dZ'$ also leads to $\langle dB\,  \Psi ,dB\, \Psi \rangle = |\alpha |^2 (dZ)^2+|\beta|^2 (dZ')^2$ from which we deduce (\ref{eq:id_n}). Similarly, the next two identities follow from  $\langle dB^\ast\,  \Psi ,dB^\ast\, \Psi \rangle \equiv (n+1) dt$ and $\langle dB^\ast \,  \Psi ,dB\, \Psi \rangle \equiv m \, dt$.

Using the fact that $\delta = - \beta$, we may subtract (\ref{eq:id_n}) from (\ref{eq:id_n+1}) to get $ 1= (2n+1 + \text{Re} \, m) (|\gamma |^2 - |\alpha |^2) = (2n+1 + \text{Re} \, m) (1-2 \text{Re} \, \alpha )$ where we used $\gamma = 1 -\alpha$. Therefore, we deduce $\text{Re}\, \alpha = \frac{n+  \text{Re}\, m}{2n+1 + 2\, \text{Re} \, m}$. 

Likewise, $\text{Im} \, (\gamma^\ast \alpha ) = \text{Im} (1-\alpha^\ast)\alpha = \text{Im} \, \alpha$. Therefore, taking the imaginary part of (\ref{eq:id_m}) leads to $\text{Im} \, m = (2n+1 +2\, \text{Re} \, m) \, \text{Im} \, \alpha$ and so
\begin{eqnarray}
    \text{Im} \, \alpha= - \text{Im} \, \gamma = \frac{\text{Im} \, m}{(2n+1 +2\, \text{Re} \, m)}.
\end{eqnarray}
Putting these together leads to the expressions (\ref{eq:alpha+gamma}).
\end{proof}

\bigskip

A pivotal result is the following one concerning conditional expectations onto von Neumann algebras: to the best of our knowledge, this first appeared as Lemma 6.1 in the (as yet unpublished) pre-print of Bouten and van Handel\cite{BvH}. 

\begin{lemma}[Quantum Bayes Formula, \cite{BvH}]
    Let $(\mathfrak{A}, \varphi)$ be a quantum probability space consisting of a von Neumann algebra and a normal state. We fix a commutative subalgebra $\mathfrak{C}$ and an element $V\in \mathfrak{C}'$, the commutant of $\mathfrak{C}$ in $\mathfrak{A}$, satisfying $\varphi (V^\ast V) =1$. An induced state $\omega$ on $\mathfrak{C}'$ is then defined by $\omega (X)=\varphi (V^\ast X V)$. Then the conditional expectation from $\mathfrak{C}'$ onto $\mathfrak{C}$ for reference state $\omega$ is given by
    \begin{eqnarray}
        \omega (X | \mathfrak{C}) = \frac{\varphi (V^\ast XV|\mathfrak{C})}{\varphi (V^\ast V|\mathfrak{C})}.
    \end{eqnarray}
\end{lemma}

We sketch the proof, as it is quite straightforward. The conditional expectation $\varphi (Y| \mathfrak{C})$ has the (least squares) property that $\varphi (\varphi (Y| \mathfrak{C})C)=\varphi (YC)$ for all $C\in \mathfrak{C}$. Setting $Y=V^\ast XV$ for given $X\in \mathfrak{C}'$ then yields $\varphi (\varphi (V^\ast XV| \mathfrak{C})C)= \varphi (V^\ast XVC) \equiv \omega (XC)$ since $[V,C]=0$. Taking the conditional expectation, this time with respect to $\omega (\cdot | \mathfrak{C})$ inside we find that $\omega (XC) = \omega ( \omega (X | \mathfrak{C})C) $ which can be rewritten in terms of the original state as $\varphi (V^\ast V \omega (X | \mathfrak{C}) C) $, noting that $V$ will also commute with $\omega (X|\mathfrak{C})$. Now, taking the conditional expectation $\varphi (\cdot | \mathfrak{C})$ under the expectation leads to $\varphi (\varphi (V^\ast V | \mathfrak{C} )\,  \omega (X | \mathfrak{C}) C) $. Therefore,
\begin{eqnarray}
    \varphi (\varphi (V^\ast XV| \mathfrak{C})C)  \equiv \varphi (\varphi (V^\ast V | \mathfrak{C} )\,  \omega (X | \mathfrak{C}) C),
\end{eqnarray} and, since $C$ was arbitrary in $\mathfrak{C}$, we deduce that
$\varphi (V^\ast XV| \mathfrak{C}) \equiv \varphi (V^\ast V | \mathfrak{C} )\,  \omega (X | \mathfrak{C})$ gives the result.

\begin{theorem}[Quantum Kallianpur-Striebel]
\label{thm:QKS}
    The quantum filter is given by
    \begin{eqnarray}
        \pi_t (X) = \frac{\sigma_t (X)}{\sigma_t(I)}
        \label{eq:QKS}
    \end{eqnarray}
    where $\sigma_t (X) = U_t^\ast \mathbb{E}[ V_t^\ast X V_t | \mathfrak{Z}_t ] U_t$ with $V_t \in \mathfrak{Z}_t'$ satisfying the quantum stochastic differential equation
    \begin{eqnarray}
        dV_t =  \bigg\{ \big(\gamma L -\alpha L^\ast \big)\otimes dZ_t  + \big( \delta L- \beta  L^\ast \big) \otimes dZ'_t  +K \otimes dt \bigg\}) V_t ,
        \label{eq:dV}
    \end{eqnarray}
    with $V_0 = I\otimes I$.
\end{theorem}

\begin{proof}
    The form (\ref{eq:QKS}) follows from the quantum Bayes formula once we can construct $V_t \in \mathfrak{Z}_t'$ with the property that $\mathbb{E}[j_t(X)]= \mathbb{E}[V_t^\ast (X \otimes) V_t]$. We now give this construction.

    Employing the double-Fock representation, the overall state comes from the vector state $\phi \otimes \Psi$ and we note that
    \begin{eqnarray}
        dU_t \, \phi \otimes \Psi \equiv
        \bigg( \sqrt{n+1}L\otimes dB_t^\ast - \sqrt{n} L^\ast \otimes dB_t +K \otimes dt \bigg) U_t \, \phi \otimes \Psi,
        \label{eq:QSDE_squeeze}
    \end{eqnarray}
with $K$ given by (\ref{eq:K}). We now note that $    dB_t \, \Psi = \alpha \, dZ_t \, \Psi +\beta \,dZ'_t \, \Psi$ and $ dB_t^\ast \, \Psi =  \gamma \, dZ_t \, \Psi + \delta \,dZ'_t \, \Psi $, and making these replacements in (\ref{eq:QSDE_squeeze}) leads to
\begin{eqnarray}
        dU_t \, \phi \otimes \Psi =
        \bigg\{ \big(\gamma L -\alpha  L^\ast \big)\otimes dZ_t  + \big( \delta L- \beta  L^\ast \big) \otimes dZ'_t  +K \otimes dt \bigg\} U_t \, \phi \otimes \Psi .
        \label{eq:QSDE_Z}
    \end{eqnarray}
It is now clear that $\mathbb{E}[j_t (X)] = \langle \phi \otimes \Psi , U_t^\ast (X \otimes I ) U_t \, \phi \otimes \Psi \rangle \equiv \langle \phi \otimes \Psi , V_t^\ast (X \otimes I ) V_t \, \phi \otimes \Psi \rangle$, as desired.
\end{proof}

\begin{corollary}(Quantum Zakai Equation)
\label{Cor_QZE}
Under the standing assumptions, the unnormalized filter $\sigma_t (X) $ satisfies
    \begin{eqnarray}
        d\sigma_t (X) = \sigma_t (\mathcal{L} X) \, dt
        + \bigg\{  \sigma_t ( X\tilde L+\tilde L^\ast X) \bigg\} dY_t ,
        \label{eq:qZakai}
    \end{eqnarray}
where $\tilde L = \gamma L - \alpha L^\ast$, or explicitly
\begin{eqnarray}
    \tilde L =L - \frac{n+   m}{2n+1+2\, \text{Re}\, m} (L+L^\ast ).
    \label{eq:tilde_L}
\end{eqnarray}
We also note the useful identity
\begin{eqnarray}
    \tilde L + \tilde L^\ast =L+L^\ast .
    \label{eq:theLs}
\end{eqnarray}
\end{corollary}
\begin{proof}
    We observe that
    \begin{eqnarray}
        V_t^\ast (X \otimes I) V_t &=& X\otimes I\nonumber \\
        &+& \int_0^t
        V_s^\ast \big[ X \big( \gamma L- \alpha L^\ast \big) +(\gamma^\ast L^\ast -\alpha^\ast L)X \big]\otimes I \, V_s \otimes dZ_s \nonumber \\
        &+& \int_0^t
        V_s^\ast \big[ X( \delta L-\beta L^\ast) (\delta^\ast L^\ast - \beta^\ast L) X\big]  \otimes I \, V_s \otimes dZ_s^\prime \nonumber \\
        && + \int_0^t V_s \big( \mathcal{M} X \big) \otimes I \big) V_s\otimes ds
    \end{eqnarray}
    where
    \begin{eqnarray}
        \mathcal{M} X &=& K^\ast X+XK \nonumber \\
        && + (2n+1+2 \, \text{Re}\, m) \,(\gamma^\ast L^\ast -\alpha^\ast L)X(\gamma L - \alpha L^\ast )
        \nonumber \\
        &&+ (2n+1+2 \, \text{Re}\, (e^{2
        i\lambda} m) ) \, (\delta^\ast L^\ast -\beta^\ast L)X (\delta L - \beta L^\ast)
        .
    \end{eqnarray}
    Note that we used the fact that $dZdZ' \equiv 0$ under our standing assumptions. 
    After a routine expansion and collection of terms, one finds from Proposition \ref{prop:alpha} - especially the identities  (\ref{eq:id_n},\ref{eq:id_n+1},\ref{eq:id_m}) - that $\mathcal{M}$, in fact, just equates to the original squeezed Lindblad generator $\mathcal{L}$ from (\ref{eq:L}). We omit the derivation as the result is somewhat obvious given that we are essentially computing the $dt$-component from the master equation by other means.

    We now take the conditional expectation onto $\mathfrak{Z}_t$ and obtain
    \begin{gather}
        \mathbb{E} [ V_t^\ast (X \otimes I) V_t  | \mathfrak{Z}_t] 
       = X\otimes I \nonumber \nonumber \\
     + \int_0^t
        \mathbb{E} \big[ 
        V_s^\ast \big\{  X \big( \gamma L- \alpha L^\ast \big) +(\gamma^\ast L^\ast -\alpha^\ast L)X \big\}\otimes I \, V_s | \mathfrak{Z_s}\big] \otimes dZ_s \nonumber \\
    +  \int_0^t
        \mathbb{E} \big[ 
        V_s^\ast \big[ \mathcal{L}X \otimes I \big]V_s | \mathfrak{Z_s}\big] \, ds
    \end{gather}
    using the fact that the increments $dZ_t'$ will have conditional mean zero.
    This implies that
    \begin{eqnarray}
        d\,  \mathbb{E} [ V_t^\ast (X \otimes I) V_t  | \mathfrak{Z}_t] 
        &=& \mathbb{E} \big[ 
        V_t^\ast \big[ X \big( \gamma L- \alpha L^\ast \big)\otimes I \big]V_t | \mathfrak{Z_t}\big] \otimes dZ_t \nonumber \\
        &&+ \mathbb{E} \big[ 
        V_t^\ast \big[ \big( \gamma^\ast L^\ast- \alpha^\ast L \big)X\otimes I \big]V_t | \mathfrak{Z_t}\big] \otimes dZ_t \nonumber \\
        &&+  
        \mathbb{E} \big[ 
        V_t^\ast \big[ \mathcal{L}X \otimes I \big]V_t | \mathfrak{Z_t}\big] \, dt .
    \end{eqnarray}
    A unitary conjugation yields $\sigma_t (X) = U_t^\ast \mathbb{E} [ V_t^\ast (X \otimes I) V_t  | \mathfrak{Z}_t] U_t$ and we note that the same conjugation converts $Z_t$ into $Y_t$.
\end{proof}

\begin{theorem}(Quantum Stratonovich-Kushner Equation)
    The (normalized) quantum filter $\pi_t (X)$ satisfies the equation
    \begin{eqnarray}
        d \pi_t (X) &=& \pi_t (\mathcal{L }X) \, dt
        +\bigg\{ \pi ( X\tilde L+ \tilde L^\ast X) -\pi_t (X) \pi_t (\tilde L+\tilde L^\ast ) \big]
        \bigg\} \, dI_t
        \label{eq:SK}
    \end{eqnarray}
    where the innovations process $I_t$ is defined by $I_0=0$ and
    \begin{eqnarray}
        dI_t = dY_t -   \pi_t ( L+ L^\ast )\, dt .
    \label{eq:innovations}
    \end{eqnarray}
\end{theorem}

\begin{proof}
    We first observe that the normalization factor satisfies $d\sigma_t (I) =  \sigma_t (\tilde L+\tilde L^\ast ) \, dY_t$. From the Ito formula and (\ref{eq:dY}), we see that
    \begin{eqnarray}
        d \frac{1}{\sigma_t (I)} = -\frac{\sigma_t (\tilde L+\tilde L^\ast ) }{\sigma_t (I)^2}  \, dY_t + (2n+1+2 \text{Re}\, m )\,\frac{ \sigma_t (\tilde L+\tilde L^\ast)^2 }{\sigma_t (I)^3} \, dt ,
    \end{eqnarray}
    therefore, from the Ito product rule
    \begin{eqnarray}
        d \pi_t (X) = d \frac{\sigma_t (X)}{\sigma _t (I)}
        = d\sigma_t (X) \, \frac{1}{\sigma _t (I)}
        +\sigma_t (X)\, d\frac{1}{\sigma _t (I)}
        +d\sigma_t (X)\, d\frac{1}{\sigma _t (I)}
    \end{eqnarray}
    and the desired result (\ref{eq:SK}) readily follows. Note that we obtain the innovations as $ dI_t = dY_t - (2n+1+2 \text{Re} \, m )\,  \pi_t (\tilde L+\tilde L^\ast )\, dt$ and the result follows from the identity (\ref{eq:theLs}).
\end{proof}

\begin{remark}
    Similar to the vacuum case, the innovations process here will have the statistics of a Wiener process, though with the variance of $I_t$ being $(2n+1 +\text{Re} \, m)\, t$.
\end{remark}
\begin{remark}
    In the thermal case ($m=0, e^{i\lambda}=\frac{1}{i}$), we have $\alpha = \frac{n}{2n+1}, \gamma = \frac{n+1}{2n+1}$, and so $\tilde L \to \frac{n+1}{2n+1} L- \frac{n}{2n+1} L^\ast$ which recovers the thermal filter results from \cite{Gough25}.
\end{remark}

\section{Discussion and Conclusion}
We have derived the quantum filter for the case where the external input fields are in a squeezed state and we conduct a homodyne detection on a quadrature of the output field. 
Our analysis is restricted to a convenient class where the second field is chosen so as to lead to what we term a balanced Bogoliubov transformation of the two vacuum field processes behind the representation.

The result has applications to quantum optics problems, but we have also been inspired by the problem of formulating a quantum open systems approach to the Unruh-DeWitt detector where the Minkowski vacuum can be written as an entangled two–mode squeezed state of left- and right-Rindler modes: see \cite{Gough25} and \cite{Scully}.  In principle, this may be applied to Hawking radiation, where we would measure the late-time output modes (the Hawking radiation field) which are entangled with their in-falling quanta behind the horizon.

The Unruh and Hawking effects stem from the common feature that different observers decompose the same quantum field into modes. One observer (Minkowski or Kruskal) may interpret the field as being in a vacuum state, however, an accelerated observer (Unruh) or a non-free-falling observer outside a black hole (Hawking) may see this as an entangled two‑mode squeezed state. In both these circumstances, when the part that is inaccessible (either lying behind the Rindler or a black‑hole horizon)  is traced out, the remaining field is thermal (no squeezing!). However, this relies explicitly on the assumption that the observer is either uniformly accelerating (Unruh) or fixed relative to a black-hole; if these conditions are dropped then the non-stationary nature should lead to reduced states will be generically squeezed. The situation is similar to dynamical Casimir effects where time-dependent boundary conditions may lead to parametric amplification and consequently observable squeezed radiation.

The mathematical filtering problem addressed here had to tackle the fact that commutant of the algebra of measured variables is, in stark contrast to the vacuum case, a nontrivial quantum field algebra. However, our calculation utilizes Araki-Woods type representations and Tomita-Takesaki formalism. We solve the problem for a single quantum input channel which is a squeezed state with parameters $(n,m)$. Crucially, the filter we must be independent of the exact choice of representation used, as otherwise the answer is unphysical; but we see that this is indeed the case here.

\section*{Acknowledgement}
The authors have the pleasant duty to thank Luc Bouten for several stimulating discussions on quantum filtering, especially with regard to the reference probability approach. We also thank Claus K\"{o}stler, Alex Belton and Martin Lindsay for raising our attention to their work on squeezed quantum noises. We are also grateful, and indeed beholden, to an anonymous referee for suggestions that not only significantly ameliorated the presentation, but alerted us to a number of algebraic errors in the original.

\end{document}